\documentclass[%
 reprint,
 amsmath,amssymb,
 aip,
]{revtex4-1}

\usepackage{xcolor}
\usepackage{graphicx}

\usepackage{amsthm}

\newtheorem{theorem}{Theorem}
\newtheorem{lemma}{Lemma}

\usepackage{subcaption}
\usepackage{dcolumn}
\usepackage{bm}
\usepackage{amsmath,comment}

\usepackage[T1,T2A]{fontenc}
\usepackage[utf8]{inputenc}
\usepackage[russian,english]{babel}

\newcommand{\secref}{Section~\ref}
\newcommand{\lemref}{Lemma~\ref}

\newcommand{\figref}[2][{}]{\figurename~\ref{#2}#1}

\newcommand*\dif{\mathop{}\!\mathrm{d}}

\newcommand{\pdif}[2]{\dfrac{\partial #1}{\partial #2}}

\begin{document}


\title{Analytical solutions for nonlinear plasma waves with time-varying complex frequency}
\author{B. J. Q. Woods}
	\affiliation{School of Mathematics, University of Leeds, Woodhouse Lane, Leeds, LS2 9JT, United Kingdom}
	\affiliation{Department of Physics, York Plasma Institute, University of York, Heslington, York, YO10 5DD, United Kingdom}

\date{\today}

\begin{abstract}
Bernstein-Kruskal-Greene (or BGK) modes are ubiquitous nonlinear solutions for the 1D electrostatic Vlasov equation, with the particle distribution function $f$ given as a function of the particle energy. Here, we consider other solutions $f = f[\epsilon]$ where the particle energy is equal to the second-order velocity space Taylor expansion of the function $\epsilon(x,v,t)$ near the wave-particle resonance. This formalism allows us to analytically examine the time evolution of plasma waves with time-varying complex frequency $\omega(t) + i \gamma(t)$ in the linear and nonlinear phases. Using a Laplace-like decomposition of the electric potential, we give allowed solutions for the time-varying complex frequencies. Then, we show that $f$ can be represented analytically via a family of basis decompositions in such a system. Using a Gaussian decomposition, we give approximate solutions for contours of constant $f$ for a single stationary frequency mode, and derive the evolution equation for the nonlinear growth of a frequency sweeping mode. For this family of modes, highly nonlinear orbits are found with the effective width in velocity of the island roughly a factor of $\sqrt{2}$ larger than the width of a BGK island.
\end{abstract}

\maketitle



\section{\label{sec:int}Introduction}
Kinetic plasmas can be described by a Boltzmann equation \cite{wesson2011tokamaks} that describes the evolution of the particle distribution function, $f$. The force exerted on each of the particles in the system manifests as a term acting as $\mathbf{F} \cdot \nabla_v$ on the distribution function, where $\nabla_v$ denotes a velocity-space gradient, and $\mathbf{F}$ is typically taken to be the Lorentz force. This leads to nonlinear coupling between the electromagnetic field and the particle distribution function. As a result of this nonlinearity, it is difficult to solve systems with kinetic plasmas analytically.

The method of linearisation allows one to derive Landau damping or inverse Landau damping of electrostatic plasma waves in the system under the assumption that the nonlinear coupling is sufficiently small \cite{landau1946vibrations}. However, this coupling leads to wave-wave instabilities which affect the nonlinear stability of the system. This can lead to destabilisation \cite{rosenbluth1997theory,lesur2016nonlinear} as well as stabilisation \cite{lin1998turbulent,anderson2017statistical}. These nonlinear scenarios have been partially investigated analytically, however these systems are typically examined using computational means \cite{vann2003fully,lang2010nonlinear,gorelenkov2018resonance,woods2018stochastic}.

The phenomenon of solitons is widely documented \cite{korteweg1895change,zabusky1965interaction,emplit1987picosecond}, where nonlinear structures form due to spatial coupling between waves in the system. In kinetic plasmas, nonlinear coupling in momentum space allows for structures to form on the distribution function known as `holes' and `clumps', referring to a local decrease or increase respectively on the spatially averaged distribution function \cite{berk1997spontaneous}. These structures propagate in a non-dispersive manner similar to solitons. However, as they move through momentum space they draw free energy from the system even when at constant amplitude.

\begin{figure}[h!!]
    \includegraphics[width=0.45\textwidth]{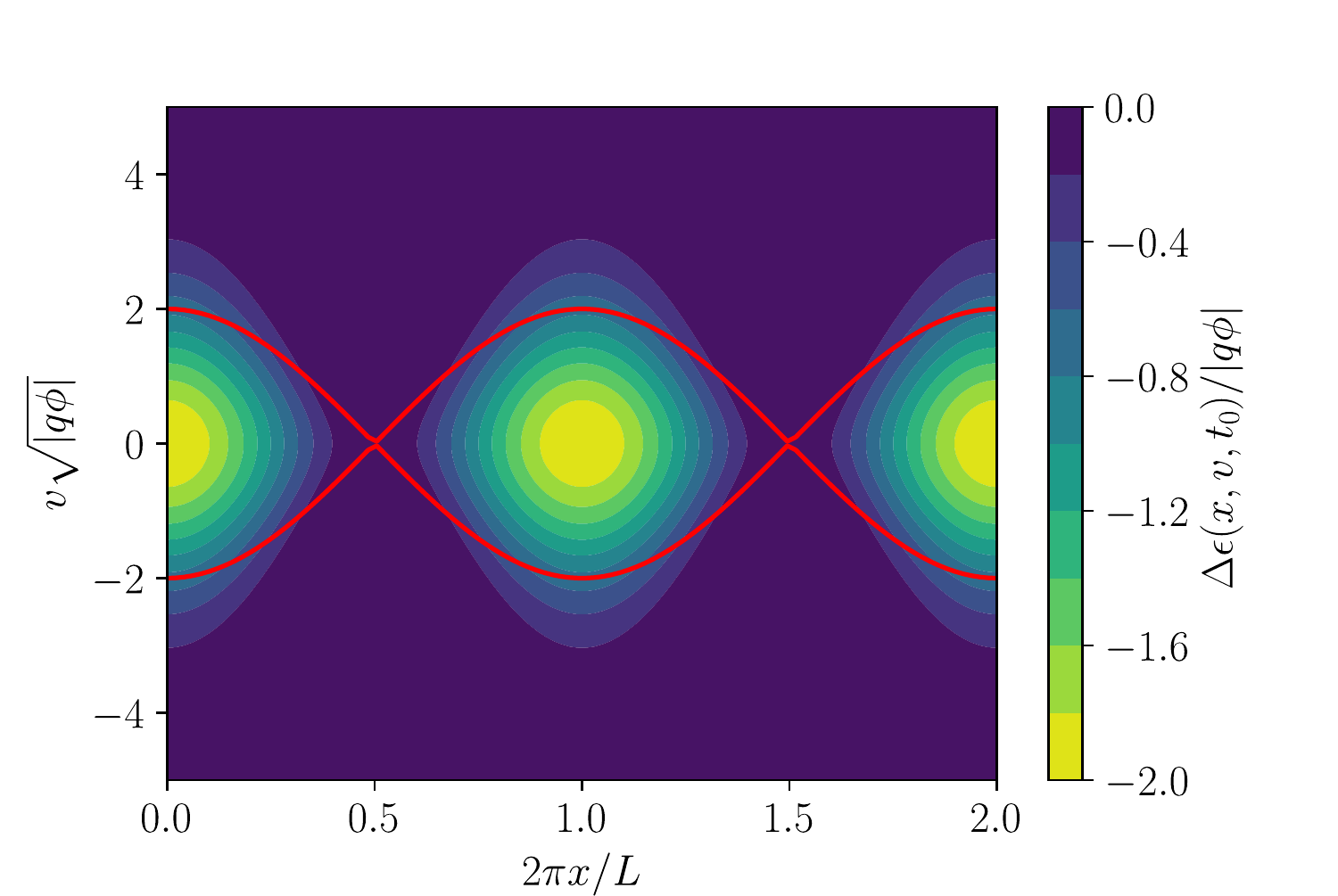}
    \caption{Contours of constant $\Delta \epsilon(x,v,t)$ for a single electric potential wave with time-invariant frequency and constant amplitude undergoing a form of resonant interaction in a 1D kinetic system (derived in \secref{sec:static}). Particles in this mode are approximately trapped within an island of width of $\sim \sqrt{4U/m}$ in velocity - a factor of $\sqrt{2}$ larger than that which is expected from BGK theory \cite{bernstein1957exact} -  where $U = |q\phi_{\textrm{max.}}|$ is the electric potential energy, and $m$ is the particle mass.}
    \label{fig:linear}
\end{figure}

Understanding how these structures are formed may prove crucial for mitigating fast ion loss in tokamaks, as well as other types of kinetic instability. In tokamaks, wave-particle resonance allows for free energy exchange between particles and the electromagnetic field, depending on momentum space gradients. As holes and clumps form, they perturb the local distribution function, allowing for destabilisation of other waves in the system. This cascading effect may play a role in the observation of abrupt large events or mode avalanching in tokamaks \cite{sharapov2002alfven,gryaznevich2006perturbative,podesta2009experimental}, where large spikes in the intensity of magnetic perturbations over a broadband of frequencies can suddenly occur. This is correlated with large amounts of fast ion loss, sharply degrading operational performance of the fusion device.

The Vlasov equation \cite{vlasov1938vibration} has been solved in previous work for the case where the particles follow orbits on which their energy is conserved. \cite{bernstein1957exact}. These types of plasma wave, known as Bernstein-Green-Kruskal (or BGK) modes, result in periodic orbits in phase-space `islands' where trapped particles resonate with the wave. Further work by Berk, Breizman and others \cite{berk1995numerical,berk1996nonlinear,berk1997spontaneous} has explored the movement of holes and clumps through phase-space by assuming that the hole and clump evolve via a temporal continuum of BGK modes. While this work has proven very successful for predicting the frequency sweep associated with these structures, the work does not fully explain how these structures form, and assume a constant slope on the distribution function during structure formation. In addition, this theory is limited to systems near marginal stability; rapidly growing waves and rapidly sweeping waves break the required adiabaticity conditions.

In this vein, work by Lilley \emph{et al.} has shown computationally how holes and clumps form via phase mixing of the wave within a phase-space island \cite{lilley2010effect}. Work by Wang \emph{et al.} has recently examined frequency chirping for the late-time evolution of plasma waves using a mixture of analytical theory and the CHIRP code \cite{wang2018frequency}. However, the threshold for hole-and-clump formation and the formation of holes and clumps on distribution functions with large curvature are both still relatively unexplored. As a result, understanding of hole-and-clump destabilisation is limited, and associated stability thresholds have not been fully investigated.

In \secref{sec:eom}, we review basic properties of the distribution function and its moments. In \secref{sec:evo}, we derive a Krook-like collision operator, and link this exact form to the complex linear dispersion relation for waves near marginal stability in the literature \cite{berk1997spontaneous}. We show that a similar dispersion relation for modes undergoing Landau resonance can be captured in a heuristically similar collisionless model.

In \secref{sec:nonlinear}, we formalise a nonlinear spectral decomposition which allows for nonlinear damped dynamics for the electrostatic potential. Then, in \secref{sec:basis} we show a general family of exact solutions for the Vlasov equation for waves with time-varying complex frequency, defined by a non-holonomic constraint on the system.

In \secref{sec:effective}, we show that it is possible to restrict the family solutions further such that $f = f[\epsilon]$, where the particle energy is equal to the second-order velocity space Taylor expansion of the function $\epsilon(x,v,t)$ near a wave-particle resonance at $v = u$. Here, the single particle orbits are such that the effective potential is given by $\phi$ plus a velocity dependent term (which is identically zero in BGK models), allowing for velocity space `screening' of the resonance.

Later in \secref{sec:effective}, we represent $f$ as a function of the energy-like function $\epsilon$, then show via a suitable basis expansion that one can easily yield approximate solutions for static frequency waves and chirping waves, with an `effective' width in $v$ of the island a factor of $\sqrt{2}$ larger than the BGK island width. Finally, we show that it is possible to construct a `first-order' distribution function that accurately reproduces every moment of the distribution function.

\section{\label{sec:eom}The distribution function}
The evolution of the single species particle distribution function is given by the 1D electrostatic Boltzmann equation with the force field given by the Lorentz force \cite{wesson2011tokamaks}:

\begin{equation}
\label{eq:boltzmann}
\pdif{f}{t} + v \pdif{f}{x} - \dfrac{q}{m} \pdif{\phi}{x} \pdif{f}{v} = \hat{C}f
\end{equation}

where $\hat{C}$ denotes a differential operator acting on the particle distribution function $f(x,v,t)$, $x \in \mathbb{R}$ is position, $v \in \mathbb{R}$ is velocity, $t \in \mathbb{R}_{\geq 0}$ is time, $q$ is the particle charge, $m$ is the particle mass, and $\phi(x,t)$ is the electric potential. The $\hat{C}$ operator encapsulates particle collisions, and other terms which allow for conservation of the rate of change of every moment of the distribution function. 

\subsection{Properties of the distribution function}
\subsubsection{Vanishing at infinity}
The scalar field $f(x,v,t)$ is required to vanish as $|v| \to \infty$:

\begin{equation}
\label{eq:vanishing}
\lim_{v \to \pm \infty} f(x,v,t) \equiv 0
\end{equation}

\subsubsection{Integrable}
The scalar field $f(x,v,t)$ is integrable on its submanifolds $x \in \mathbb{R}$ and $v \in \mathbb{R}$:

\begin{equation}
\label{eq:integrable}
\left|\int\limits_{C} f(x,v,t) \dif \ell\right| < \infty
\end{equation}

where $C$ is a contour that spans two points on $\mathbb{R} \times \mathbb{R}$.

\subsubsection{Moments}
The $l^{\textrm{th}}$ moment of the scalar field $S(v)$ is defined as:

\begin{equation}
M_l[S] \equiv \int\limits_{-\infty}^{\infty} v^{l} S(v) \dif v
\end{equation}

Each and every moment $l \geq 0$ of $f(x,v,t)$ is finite, that is to say:

\begin{equation}
|M_l[f]| < \infty \, \forall \, l \geq 0
\end{equation}

We list the following useful moments of the distribution function:

\begin{subequations}
\begin{equation}
M_0[f] \equiv n(x,t) 
\end{equation}
\begin{equation}
M_1[f] \equiv \dfrac{1}{q} J(x,t) 
\end{equation}
\begin{equation}
M_2[f] \equiv \dfrac{2}{m} T(x,t)
\end{equation}
\end{subequations}

where $n(x,t)$ is the particle density, $J(x,t)$ is the current density, and $T(x,t)$ is the kinetic energy density.

\section{Evolution equations}
\label{sec:evo}
The evolution of the electric potential is given here by the Maxwell-Amp\`{e}re law:

\begin{equation}
\label{eq:ma_law}
0 = \mu_0 \int\limits_{-\infty}^{\infty} qv f \dif v - \mu_0 \epsilon_0 \pdif{}{t} \pdif{\phi}{x}
\end{equation}

where we have taken ${B} = 0$.

\subsection{Energy balance}
By taking the second moment of equation \eqref{eq:boltzmann}:

\begin{equation}
\int\limits_{-\infty}^{\infty} \left[ \dfrac{1}{2} mv^2 \pdif{f}{t} + qvf\pdif{\phi}{x}\right] \dif v \equiv \left( \pdif{T}{t} \right)_{\textrm{C,T}}
\end{equation}

where $(\partial T / \partial t)_{\textrm{C,T}}$ allows for sources and sinks of energy density:

\begin{equation}
\label{eq:sink}
\left( \pdif{T}{t} \right)_{\textrm{C,T}} \equiv \dfrac{1}{2} m \int\limits_{-\infty}^{\infty} v^2 (\hat{C}'f) \dif v 
\end{equation}

where $\hat{C}'f \equiv \hat{C}f - v\partial f/\partial x$ encapsulates collisions as well as spatial advection. Use of equation \eqref{eq:ma_law} allows one to find an energy conservation equation:

\begin{equation}
\label{eq:energy_balance}
\int\limits_{-\infty}^{\infty} \left[ \dfrac{1}{2} mv^2 \pdif{f}{t} \right] \dif v + \dfrac{\epsilon_0}{2} \pdif{}{t} \left(\pdif{\phi}{x}\right)^2 = \left( \pdif{T}{t} \right)_{\textrm{C,T}}
\end{equation}

where the first term denotes the rate of change of energy density in the distribution function, the second term denotes the rate of change of electric field energy density, and the third term denotes sources and sinks of energy density.

\subsection{Berk-Breizman sink}
Berk and Breizman\cite{berk1995numerical} use an energy sink which dissipates plasma waves at a rate $\gamma_D$. This allows for a region of nonlinear stability for the system in the presence of instabilities.

One can represent this sink as the following:

\begin{equation}
\label{eq:bb_sink}
\left( \dfrac{\partial T}{\partial t} \right)_{\textrm{C,T}} = -\gamma_D \dfrac{\epsilon_0}{2} \left(\pdif{\phi}{x}\right)^2
\end{equation}

Therefore, by using equations \eqref{eq:sink}, \eqref{eq:energy_balance} and \eqref{eq:bb_sink} one finds that for a vanishing integrand under $v$:

\begin{equation}
\pdif{}{t} (\hat{C}'f) = -\gamma_D \left[(\hat{C}'f) - \pdif{f}{t}\right]
\end{equation}

This is an inhomogeneous decay equation of the form $y(t)' = -ay + b(t)$. The solution, given by integrating factor, is:

\begin{equation}
\hat{C}'f = e^{-\gamma_D t} \left[\left.\pdif{f}{t}\right|_{t=0} + \gamma_D \int\limits_{0}^{t} e^{\gamma_D \tau} \pdif{f}{\tau} \dif \tau\right]
\end{equation}

If one assumes the contribution from heat diffusion to be negligible in comparison to the contribution from collisions and other sources and sinks, then a Maclaurin expansion in $\gamma_D$ yields:

\begin{equation}
\hat{C}f \approx \left.\pdif{f}{t}\right|_{t=0} e^{-\gamma_D t} + \gamma_D(f - F) + \mathcal{O}(\gamma_D^2)
\end{equation}

where $F \equiv f(x,v,0)$, such that a Krook-like collision operator safely approximates the behaviour to first order in $\gamma_D$. The higher order terms also vanish for early times where $\gamma_D t \ll 1$.

\subsection{Complex linear dispersion relation}
We spectrally decompose the distribution function and other spatially dependent quantities as follows:

\begin{equation}
f(x,v,t) = \dfrac{1}{2} \sum\limits_j \left[f_j(v,t) e^{\textrm{i}k_j x} + \textrm{c.c.}\right]
\end{equation}

where $\{k_j\}$ is a set of wavenumbers, and $\textrm{c.c.}$ denotes the complex conjugate. Using the Berk-Breizman sink, it is possible to linearise equations \eqref{eq:boltzmann} and \eqref{eq:ma_law} to produce the following complex linear dispersion relation for a single mode, under the Krook-like approximation:

\begin{equation}
p|_{t = 0} \approx \dfrac{q^2}{m \epsilon_0} \int\limits_{\Omega} \left.\pdif{f_0}{v}\right|_{t = 0} \, \dfrac{v \dif v}{p|_{t = 0} - \gamma_D + \textrm{i} k_j v}
\end{equation}

where $\Omega$ is the suitable Landau contour for the problem, and $p = \gamma - \textrm{i} \omega$. The linear growth rate is given by $\gamma_L = \gamma(t=0)$, and $\omega(t)$ is the frequency of the mode.

Commonly in the literature, the `modified' version of the Maxwell-Amp\`{e}re law is used instead to generate the same energy balance equation as equation \eqref{eq:energy_balance}:

\begin{equation}
\pdif{}{t} \pdif{\phi_{\textrm{alt.}}}{x} = \dfrac{q}{\epsilon_0} \int\limits_{-\infty}^{\infty} v f \dif v - \gamma_D \pdif{\phi_{\textrm{alt.}}}{x}
\end{equation}

where $\phi_{\textrm{alt.}}$ is a modified electric potential. For these models to be consistent, Maxwell's equations must retain the same canonical form; previous work shows that transformation of the system Lagrangian that preserves the system Hamiltonian is indeed possible under certain conditions. However, for all counterexamples, these models do not form a fully consistent set of equations. \cite{woods2018stochastic}

Using this alternative model, one obtains a very similar complex linear dispersion relation \cite{woods2018stochastic}: 

\begin{equation}
p_{\textrm{alt.}}|_{t = 0} + \gamma_D = \dfrac{q^2}{m \epsilon_0} \int\limits_{\Omega} \left.\pdif{f_0}{v}\right|_{t = 0} \, \dfrac{v \dif v}{p_{\textrm{alt.}}|_{t = 0} + \textrm{i} k_j v}
\end{equation}

This means that while the model is inconsistent (energy is typically not conserved), it produces similar dynamics to a model with Krook-like collisions; the primary difference is that the location of the Landau resonance is shifted from $p = \gamma_D - \textrm{i} k_j v$ to $p = -\textrm{i} k_j v$.

In the nonlinear phase of the systems evolution, it is a well known result that the original resonance undergoes broadening. Therefore, provided that $\gamma_D/k_j$ is much smaller than the resonance width, these models perform well except for energy conservation; as we have already utilised a small $\gamma_D$ approximation to assume Krook-like dissipation, we are safe to make this assumption.

For this reason, in the rest of the paper we will explore $\hat{C}f = 0$, such that the original and `modified' versions of the Maxwell-Amp\`{e}re law are the same. We expect that an extension to this work featuring a Berk-Breizman sink such that $\hat{C}f \neq 0$ could be investigated approximately by still considering a collisionless kinetic equation, and considering a model using the `modified' Maxwell-Amp\'{e}re law. This would allow one to retain the formalism derived here for the collisionless case, provided that the damping is small.

\section{Nonlinear damped dynamics}
\label{sec:nonlinear}
\subsection{Complex frequency evolution}
By coupling the Maxwell-Amp\`{e}re law to the Boltzmann equation, we obtain a system of differential equations which features strong nonlinearities. By taking the time derivative of equation \eqref{eq:ma_law}:

\begin{equation}
\pdif{^2}{t^2} \left(\pdif{\phi}{x}\right) = \dfrac{q}{\epsilon_0} \int\limits_{-\infty}^{\infty} v \pdif{f}{t} \dif v
\end{equation}

From here, one can show that substituting $\partial_t f$ for equation \eqref{eq:boltzmann} gives an equation for $(\partial_x \phi)$ similar to driven simple harmonic motion:

\begin{equation}
\label{eq:forced_shm}
\renewcommand{\arraystretch}{2.8}
\begin{array}{r l}
\pdif{^2}{t^2} \left(\pdif{\phi}{x}\right) &= - \overbrace{\dfrac{q^2}{m \epsilon_0} n}^{\textrm{oscillation}} \left(\pdif{\phi}{x}\right) \\
& \hspace{30pt} + \underbrace{\dfrac{q}{\epsilon_0} \int\limits_{-\infty}^{\infty} v (\hat{C}'f) \dif v}_{\textrm{drive}}
\end{array}
\end{equation}

However, one should note that $f = f(x,v,t)$, and therefore this is only part of a coupled system of equations. We shall solve equation \eqref{eq:forced_shm} by representing the electric potential $\phi$ via a set of wavepackets $\{\phi_l\}$:

\begin{equation}
\phi(x,t) \equiv \sum\limits_{l} \phi_l(t)
\end{equation}

Then, we employ the following solution for each wavepacket:

\begin{equation}
\label{eq:wavepacket}
\begin{array}{r l}
\phi_l(x,t) &\equiv \dfrac{1}{2} \displaystyle\sum\limits_s \Bigg\{\phi_{sl} \exp \left[ i k_l x + \int\limits_{0}^{t} p_{sl}(\tau) \dif \tau\right] \\
&\hspace{50pt} + \textrm{ c.c.} \Bigg\}
\end{array}
\end{equation}

where $\{k_l\}$ is a set of wavenumbers, $\{\phi_{sl}\}$ is a set of wave amplitudes, and $\{p_{sl}\}$ is a set of complex frequencies:

\begin{equation}
p_{sl}(t) = \gamma_{sl}(t) - i \omega_{sl}(t)
\end{equation}

where $\{\gamma_{sl}(t)\}$ is a set of nonlinear growth rates, and $\{\omega_{sl}(t)\}$ is a set of constituent frequencies. Using this solution, one finds the following equation:

\begin{equation}
\label{eq:gsl}
\dfrac{\dif p_{sl}}{\dif t} = -p_{sl}^2 + \Gamma_{sl}^2
\end{equation}

\clearpage

\onecolumngrid

\begin{figure}[t!]
    \centering
    \begin{subfigure}[t]{0.45\textwidth}
        \centering
        \includegraphics[width=\textwidth]{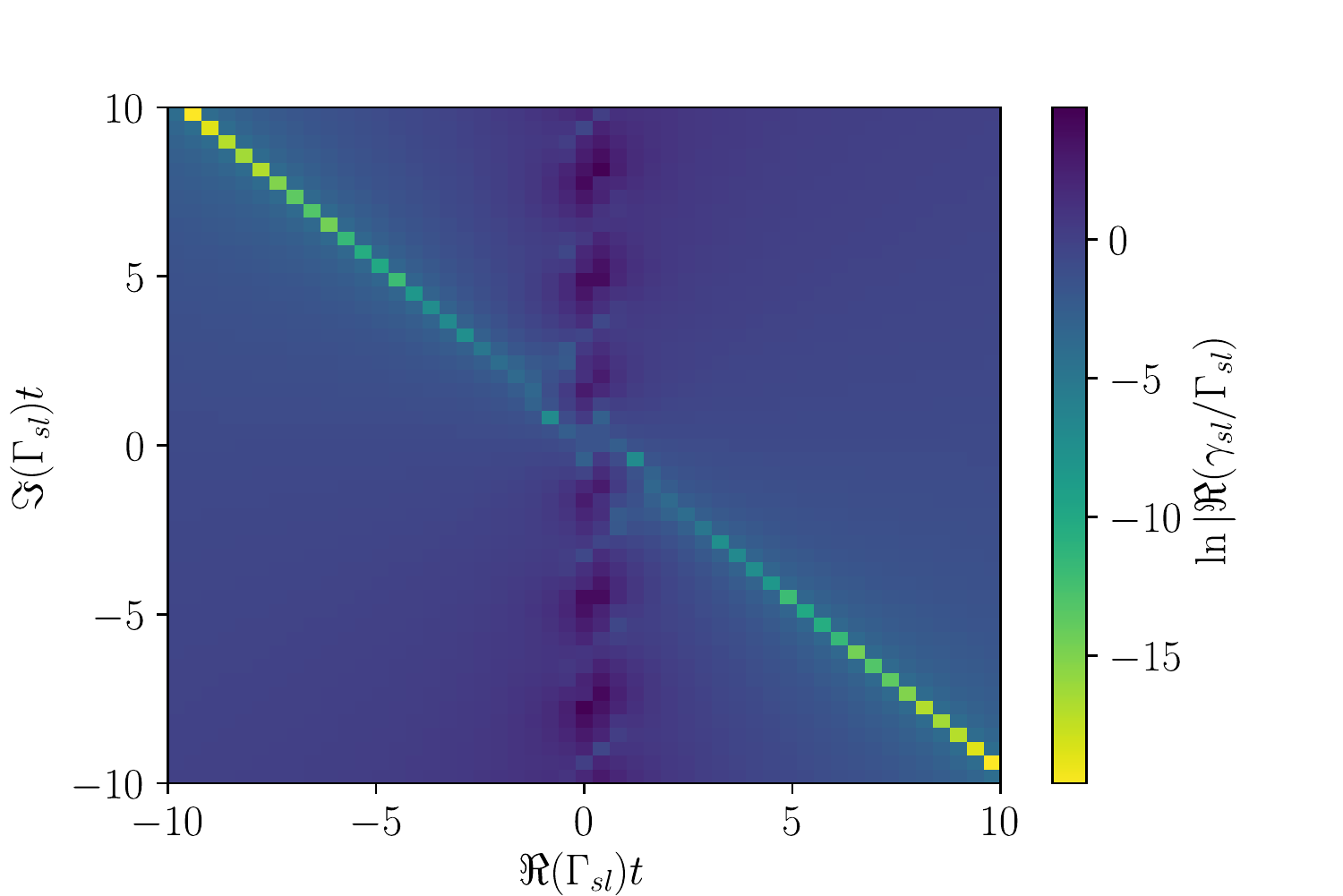}
        \caption{The growth rate is singular at $\Re(\Gamma_{sl}) = 0$ for $t \neq 0$. The condition $\Re(\Gamma_{sl}) = -\Im(\Gamma_{sl})$ leads to $\gamma_{sl}$ = 0. The magnitude of the growth rate asymptotically approaches $|\Re(\Gamma_{sl})|$.}
    \end{subfigure} ~ 
    \begin{subfigure}[t]{0.45\textwidth}
        \centering
        \includegraphics[width=\textwidth]{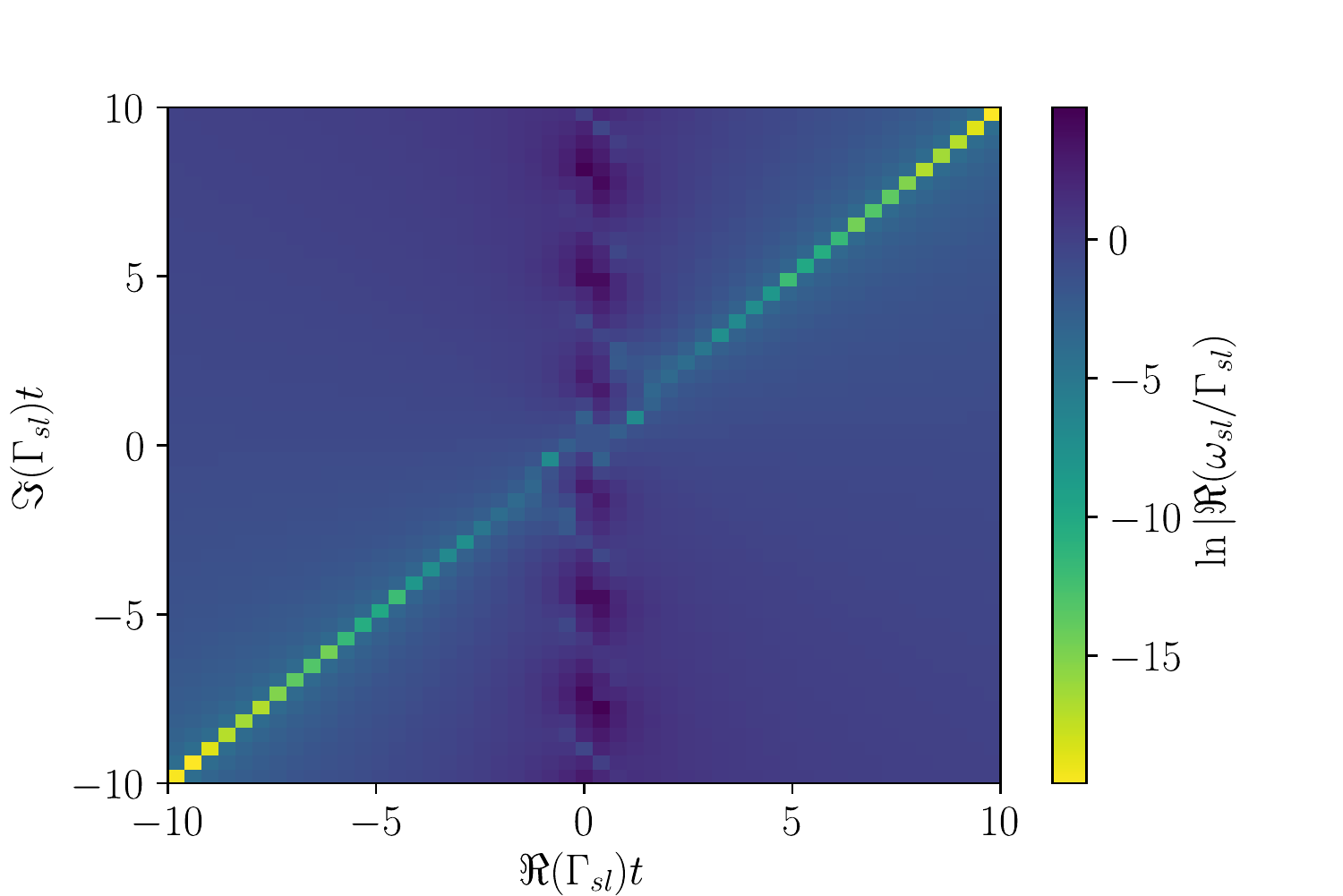}
        \caption{The frequency is singular at $\Re(\Gamma_{sl}) = 0$ for $t \neq 0$. The condition $\Re(\Gamma_{sl}) = \Im(\Gamma_{sl})$ leads to $\omega_{sl}$ = 0. The magnitude of the frequency asymptotically approaches $|\Im(\Gamma_{sl})|$.}
    \end{subfigure} 
    \caption{Logarithmic plots of the growth rate $\omega_{sl}$ and angular frequency $\gamma_{sl}$ for waves in a single species system with $\omega(t=0) = \gamma(t=0) = 0$. The complex constant $\Gamma_{sl}$ is related to the initial system conditions for a wave with wavenumber $k_l$. The $x$-axis features the real part of $\Gamma_{sl}$, while the $y$-axis features the imaginary part of $\Gamma_{sl}$. Evolution of the system can be traced by following a radial path from the centre of the plots. Where $|\Re(\Gamma_{sl})t|  \leq 1$, we observe undulating behaviour in the frequency and growth rate. The amplitude of the undulation decreases exponentially with $t/|\Gamma_{sl}|$, the length of a radial trajectory from the temporal origin at $(0,0)$.}
    \label{fig:gsl}
\end{figure}

\twocolumngrid

where $\{\Gamma_{sl}\}$ is a set of constants, corresponding to rates of change that determine the dynamics of the system:

\begin{equation}
\renewcommand{\arraystretch}{2.5}
\begin{array}{r l}
\Gamma_{sl}^2 &= - \omega_{\textrm{pl.}}^2 + \dfrac{\omega_{\textrm{pl.}}^2}{\bar{n}_0 q (\partial_x \phi)_{sl}} \bigg\{\displaystyle\int\limits_{-\infty}^{\infty} \bigg[(f - F) \partial_x \phi\\
&\hspace{20pt} - m v [\hat{C}' f]_{sl} \bigg] \dif v \bigg\}_{sl}
\end{array}
\end{equation}

where $\omega_{\textrm{pl.}} = \sqrt{\bar{n}_0 q^2/m \epsilon_0}$ is the plasma frequency, $F(v) = f(x,v,t=0)$ is the distribution function at a time $t = 0$, and $\bar{n}_0$ is the spatial average of the particle density at a time $t = 0$.

As there are multiple frequencies for $\phi_l(t)$, it is instructive to note that the complex frequency of the $l^{\textrm{th}}$ mode is made up of contributions from all of the branches.

\begin{theorem}
There is no solution where $\lim_{t \to \infty} \gamma_{sl} \leq 0$ except for the case where $\gamma_{sl}$ is constant.
\end{theorem}

\begin{proof}
Equation \eqref{eq:gsl} yields:

\begin{equation}
\label{eq:psl}
p_{sl}(t) = \Gamma_{sl} \tanh \left[\varphi + \Gamma_{sl} t\right]
\end{equation}

where the complex angle $\varphi \equiv \textrm{arctanh} (p_{sl}(0) / \Gamma_{sl})$. We show plots of $p_{sl}$ in \figref{fig:gsl}. If $\Gamma_{sl}$ is equal to the initial complex frequency, then $p_{sl}$ does not change value. For the case where $\Gamma_{sl} = 0$, and $p_{sl} \neq 0$, there is no stable solution. For all other cases:

\[
\forall \, \Gamma_{sl} \not\in \{0, p_{sl}\} : \lim_{t \to \infty} p_{sl}(t) = |\Re(\Gamma_{sl})| \pm i \Im(\Gamma_{sl})
\]

where $\Re(\Gamma_{sl})$ denotes the real part of $\Gamma_{sl}$, $\Im(\Gamma_{sl})$ denotes the imaginary part of $\Gamma_{sl}$, and $\pm$ corresponds to the sign of $\Re(\Gamma_{sl})$. Therefore, no solution exists where $\lim_{t \to \infty} \gamma_{sl} \leq 0$ except for the case where $\gamma_{sl}$ is temporally static.
\end{proof}

By integrating equation \eqref{eq:psl}, one can rewrite equation \eqref{eq:wavepacket} as the following:

\begin{equation}
\phi_l(x,t) = \sum\limits_s \left\{\dfrac{\phi_{sl}}{2} \textrm{e}^{ik_l x} \cosh \left[\varphi + \Gamma_{sl} t\right] + \textrm{c.c.}\right\}
\end{equation}

We note that for the case where $\varphi \to \infty$, $p_{sl}(t) \to p_{sl}(0)$ and therefore:

\begin{equation}
\lim_{\varphi \to \infty} \cosh \left[\varphi + \Gamma_{sl} t\right] =: e^{p_{sl}(0) t}
\end{equation}

We also impose the following physical limit:

\begin{equation}
\label{eq:limit}
\exists \, \{\Gamma_{sl}\} : \lim_{t \to \infty} \phi = \textrm{const.}
\end{equation}

such that $\lim_{t \to \infty} |\phi| < \infty$. However, finding solutions that satisfy this limit is not trivial. One can decompose equation \eqref{eq:gsl} into two real equations:

\begin{subequations}
\begin{equation}
\label{eq:gamma_ev}
\dfrac{\dif \gamma_{sl}}{\dif t} = -(\gamma_{sl}^2 - \omega_{sl}^2) + \Re(\Gamma_{sl}^2)
\end{equation}
\begin{equation}
\label{eq:omega_ev}
\dfrac{\dif \omega_{sl}}{\dif t} = -2\gamma_{sl} \omega_{sl} - \Im(\Gamma_{sl}^2)
\end{equation}
\end{subequations}

These equations can be linearised from a time $t = 0$ where  $\dif \gamma_{sl}/\dif t = 0$. From equation \eqref{eq:gamma_ev}:

\[
\omega_{sl}^2(t) \approx \omega_{\text{pl.}}^2 + \left[\gamma_{sl}^2 + \omega_{\text{pl.}}^2 \lambda^2_{sl}\right]
\]

where we have implicitly discarded some nonlinear terms $\mathcal{O}[(f - F)^2]$, and where the dimensionless constant $\lambda^2_{sl}$ is given by:

\[
\lambda^2_{sl} \approx \Re \int\limits_{-\infty}^{\infty} \dfrac{\left[q(f- F) (\partial_x \phi) - mv (\hat{C}'f)\right]_{sl}}{\bar{n}_0 q (\partial_x \phi)_{sl}} \dif v
\]

As such, one can see that at $t = 0$, for $\gamma_{sl} \ll \omega_{\textrm{pl.}}$ the frequency is approximately equal to the plasma frequency.

\section{Basis expansion of the Vlasov equation}
\label{sec:basis}
The Vlasov equation is given by\cite{vlasov1938vibration}:

\begin{equation}
\pdif{f}{t} + v \pdif{f}{x} - \dfrac{q}{m} \pdif{\phi}{x} \pdif{f}{v} = 0
\end{equation}

We begin by stating:

\[
\forall \, x,v,t \, \exists \, \epsilon: f = f(\epsilon) 
\]

This new variable $\epsilon$ is defined everywhere. Let us further impose that $f(\epsilon)$ is analytic for $\epsilon \in \mathbb{R}$. Then:

\begin{equation}
\dfrac{\dif f}{\dif \epsilon} \left[\pdif{\epsilon}{t} + v \pdif{\epsilon}{x} - \dfrac{q}{m} \pdif{\phi}{x} \pdif{\epsilon}{v} \right] = 0
\end{equation}

Ignoring the trivial solution $\dif f/\dif \epsilon = 0$, we find that $\epsilon$ also satisfies the Vlasov equation:

\begin{equation}
\label{eq:epsilon_vlasov}
\pdif{\epsilon}{t} + v \pdif{\epsilon}{x} - \dfrac{q}{m} \pdif{\phi}{x} \pdif{\epsilon}{v} = 0
\end{equation}

By using this approach, one can represent $f$ as a functional of $\epsilon$, where $\epsilon$ can be chosen to have some physical significance. We use a basis decomposition in $v$ of $\epsilon$:

\begin{equation}
\epsilon(x,v,t) \equiv \sum\limits_j c_j(x,t) g_j(x,v,t)
\end{equation}

This form of the basis decomposition is key, and corresponds to a different basis decomposition at every point in $x$ and $t$. By allowing each $(x,t)$ to permit a different basis decomposition, we allow for nonlinear coupling which is typically not allowed using a variables separable method of the form $f = XT(x,t)V(v)$.

To satisfy equations \eqref{eq:vanishing} and \eqref{eq:integrable}, we desire that $\{g_j(x,v,t)\}$ are square-integrable functions on $v \in \mathbb{R}$, and to do so we make the constraint that $\{g_j(x,v,t)\}$ are compactly supported. Accordingly, $\{g_j(x,v,t)\}$ also vanish at infinity:

\begin{equation}
\lim_{v \to \pm \infty} g_j(x,v,t) = 0 \, \forall \, j
\end{equation}

Then, if we take the $0^{\textrm{th}}$ moment of the reduced Vlasov equation:

\begin{equation}
\int\limits_{-\infty}^{\infty} \left[v \pdif{\epsilon}{x} + \pdif{\epsilon}{t}\right] \dif v = 0
\end{equation}

Using our basis expansion:

\begin{equation}
\label{eq:basis1}
\begin{array}{l}
\displaystyle\sum\limits_j \bigg[\pdif{c_j}{x} M_1[g_j] + \pdif{c_j}{t} M_0[g_j] \\
\hspace{60pt}  + c_j \pdif{M_0[g_j]}{t} + c_j \pdif{M_1[g_j]}{x}  \bigg] = 0
\end{array}
\end{equation}

Next, we perform a co-moving Galilean transform to $g(x,v,t) = g_j(x,v-u_j(t),t)$, where $u_j$ is at this point an undetermined function. Then:

\[
\begin{array}{l}
\displaystyle\int\limits_{-\infty}^{\infty} v g_j \dif v = \int\limits_{-\infty}^{\infty} v g_j \dif (v-u_j) \\
\hspace{20pt} = u_j(t) \displaystyle\int\limits_{-\infty}^{\infty} g_j \dif (v-u_j) + \int\limits_{-\infty}^{\infty} (v-u_j) g_j \dif (v - u_j)
\end{array}
\]

If $g_j$ is symmetric about $v - u_j(t) = 0$, then the second integral vanishes due to antisymmetry. Then:

\begin{equation}
M_1[g_j] = u_j(t) M_0[g_j]
\end{equation}

While it appears that this is a restriction on the solutions that are allowed, in fact any $\epsilon(x,v,t)$ can be represented via such a decomposition. The form of the basis functions has been restricted, but the function it represents has not (c.f. Fourier expansion with phase angles versus Fourier expansion with sines and cosines only).

Therefore from equation \eqref{eq:basis1}, one finds:

\begin{equation}
\label{eq:basis2}
\begin{array}{l}
\displaystyle\sum\limits_j \left\{M_0[g_j] \left[u_j \pdif{c_j}{x} + \pdif{c_j}{t} \right]\right\} \\
\hspace{30pt} = -\displaystyle\sum\limits_j \left\{ c_j \left[u_j \pdif{M_0[g_j]}{x} + \pdif{M_0[g_j]}{t} \right]\right\} \\
\end{array}
\end{equation}

This takes a similar form to a fluid advection equation with a source term. 

At this point, we define an integral mapping $x,t \longmapsto \chi$, where $\chi$ is a matrix of coordinates in instantaneous comoving frames:

\begin{equation}
\chi_{jl} = k_l x - \int\limits_{0}^{t} \omega_{jl} (\tau) \dif \tau
\end{equation}

where $\chi_{jl}/k_l$ represents a Galilean transform to the comoving frame with variable wave speed $u_j \equiv \omega_{jl}(t) / k_l$. In this sense, $\chi$ contains all of the transformations possible with the available frequencies $\{\omega_j\}$ and available wavenumbers $\{k_l\}$ in the system.

Our basis decomposition is therefore in such a form that the symmetry point in $v$ of the even functions $\{g_j\}$ is the variable wave speed of a wave with frequency $\omega_{jl}$ and wavenumber $k_l$. Therefore, one can intuitively posit that such a decomposition yields phase-space structures which move with the frequency of some form of wave in the system. If these waves are chosen to be electrostatic waves in the system, then one expects that the migration of phase-space structures corresponds to nonlinear frequency sweeping moving the wave-particle resonances.

Then, from equation \eqref{eq:basis2}:

\begin{equation}
\sum\limits_{j} \left[ \left. \pdif{}{t}\right|_{\chi} \cdot \left(c_j M_0[g_j] \right) \right] = 0
\end{equation}

This forms a non-holonomic constraint on the system. This result is particularly interesting as it implies that it is not possible for one to set up a suitable Lagrangian for the system with $c_j M_0[g_j]$ as generalised fields. It is therefore our belief that this formalism cannot be derived from an action-angle formulation. In such a formulation, Hamilton's principle assumes only holonomic constraints on the system. In contrast, our formalism yields an analog of d'Alembert's principle, allowing for more exotic scenarios involving non-conservative forces.

We have shown that $\epsilon$ solves the $0^{\textrm{th}}$ moment of the Vlasov equation; by re-examining the Vlasov equation, we can now explicitly represent the mapping between the electric potential $\phi$ and $\epsilon$.

\section{Effective potential}
\label{sec:effective}
\subsection{Motivation from BGK theory}
We expect to recover the BGK solution \cite{bernstein1957exact} near a given resonance $v = u$:

\begin{equation}
\lim_{v \to u} \epsilon(x,v,t) \approx q \phi + \dfrac{1}{2} m(v-u)^2
\end{equation}

where $\phi(x,t)$ is given by:

\begin{equation}
\phi(x,t) = \dfrac{1}{2} \sum\limits_{l} \left\{\phi_l \exp \left[-i \omega_l t + i k_l x\right] + c.c. \right\}
\end{equation}

and where $u \equiv \omega_l / k_l$ for all frequencies $\{\omega_l\}$ and all wavenumbers $\{k_l\}$. The BGK solution is only for a single phase velocity. Therefore, we Taylor expand $\epsilon$ around $v = u$:

\[
c_j \left[\left.g_j\right|_{v = u_j} + \left.\pdif{g_j}{v}\right|_{v = u_j} \hspace{-20pt} (v-u_j) + \dots \right] \approx q \phi + \dfrac{1}{2} m(v-u)^2
\]

By inspection, we find the possible approximation:

\begin{equation}
\epsilon(x,v,t) \approx q \phi \left[ 1 + \dfrac{m(v-u)^2}{2 q \phi} + \cdots \right]
\end{equation}

With a form satisfying this being a Gaussian:

\begin{equation}
\epsilon(x,v,t) = q \phi \exp \left[ \dfrac{m(v-u)^2}{2 q \phi} \right]
\end{equation}

This is not a particularly robust solution: the function is not analytic in $\chi$ at $\phi(x,t) = 0$, Gaussian basis functions are not orthonormal, and $g_j$ does not have the correct growth rate. However, every moment of this function is analytic in $x$ if $\phi(x) \leq 0$:

\begin{equation}
\left|\dfrac{\dif}{\dif x} M_l \left[ \epsilon \right] \right| < \infty ,\ \forall \, l
\end{equation}

Furthermore, in this approximate form, particles enjoy an \emph{effective potential} in phase-space which is dependent on their velocity. Particles far from the resonance are `screened' in a manner similar to Debye shielding, however this phenomenon occurs in momentum space, not real space.

To illustrate this further, one can represent $\epsilon$ in the form:

\begin{equation}
\epsilon(x,v,t) = \epsilon_0 + \mathcal{O}(v-u)^3
\end{equation}

where $\epsilon_0 = q \phi + \frac{1}{2}m(v-u)^2$ is the particle energy. As particles follow phase-space orbits given by contours of constant $f$, the orbit is given by an orbit where the potential energy have effectively been modified by terms of $\mathcal{O}(v-u)^3$.

These  extra terms allow for nonlinear orbits where particles exchange differing amounts of energy with the wave depending on how close they are to the resonance; this is in better accordance with the Landau damping \cite{landau1946vibrations}.

\subsection{Distribution function construction}

\begin{lemma}
\label{lemma:distributive}
The distribution function $f$ can be represented as a functional in the \textbf{distributive form}:

\[
f[\epsilon_0 + \Delta \epsilon] = f[\epsilon_0] + \Delta f[\Delta \epsilon]
\]

where $\epsilon_0$ is an arbitrary function, and $\Delta \epsilon = \epsilon - \epsilon_0$.

\end{lemma}

\begin{proof}
First, we state that the vector $\vec{\mu} \equiv (\begin{array}{c c c}x & v & t\end{array})$. We extend our statement that $f$ is analytic in $\epsilon$ to instead require that $f$ permits a functional Taylor expansion of $f$ about a function $\epsilon_0$: \cite{dreizler2013density}

\begin{equation}
\label{eq:taylor}
f = f[\epsilon_0] + \sum\limits_l \int\limits_{\Omega^l} \left(\prod_{j=1}^l \Delta \epsilon_j \dif \vec{\mu}_j \dfrac{\delta}{\delta \epsilon_0(\vec{\mu}_j)} \right) f[\epsilon_0]
\end{equation}

where $\Omega \in \mathbb{R}^D$ is a D-dimensional closed volume in which $d\vec{\mu}_j$ is the infinitessimal volume element, $\Delta \epsilon_j  \equiv \Delta \epsilon(\vec{\mu}_j)$ and $(\delta/\delta \epsilon)$ denotes the functional derivative with respect to $\epsilon$, given by:

\begin{equation}
\int\limits_{\Omega} \dfrac{\delta f[\epsilon]}{\delta \epsilon(\vec{\mu})} \eta(\mu) \dif \vec{\mu} := \lim_{\varepsilon \to 0} \dfrac{f[\epsilon(\vec{\mu}) - \varepsilon \eta(\vec{\mu})] - f(\epsilon)}{\varepsilon}
\end{equation}

where $\varepsilon \eta(\mu)$ is equivalent to $\Delta \epsilon(\mu)$. It is then possible for us to express equation \eqref{eq:taylor} in the form:

\[
f = f[\epsilon_0] + \sum\limits_l \int\limits_{\Omega^l} \mathcal{C}_l[\epsilon_0(\mu_1),\epsilon_0(\mu_2),\dots] \prod_{j=1}^l (\Delta \epsilon_j \dif {\mu_j})
\]

where $\mathcal{C}_l$ denotes a functional of many functions. Each function has a different domain, but an identical codomain. This codomain is the Banach space defined by $\epsilon_0$. Therefore by inspection, the nonlocal functional $\Delta f$ is given by:

\[
\Delta f[\Delta \epsilon] := \sum\limits_l \int\limits_{\Omega^l} \mathcal{C}_l[\epsilon_0] \prod_{j=1}^l (\Delta \epsilon(\vec{\mu_j}) \dif {\vec{\mu_j}})
\]

\end{proof}

Then, it is possible to show that our perturbation to the distribution function $\delta f$ can be represented using an entirely different basis than that which is employed with $\epsilon_0$.

\begin{lemma}
\label{lemma:zero}
If $\exists \epsilon_0: \partial f[\epsilon_0]/\partial t = 0$, then the following Dirichlet boundary condition must be satisfied:

\[
\Delta f[\Delta \epsilon]|_{t=0} = 0
\]
\end{lemma}

\begin{proof}
The functional form of $\epsilon_0$ is solely a function of $x$ and $v$. Therefore, it is possible to enforce the stronger constraint:

\begin{equation}
\exists \, \epsilon_0: f[\epsilon_0] \equiv f(x,v,t=0)
\end{equation}

Accordingly, as $f$ is analytic in $t$, it permits a Maclaurin expansion in $t$. Expanding $f$ as such gives:

\[
\Delta f[\Delta \epsilon] = \left.\sum\limits_l \pdif{^l f}{t^l}\right|_{\epsilon = \epsilon_0} \hspace{-10pt} \dfrac{t^l}{l!}
\]

By setting $t = 0$, the aforementioned Dirichlet boundary condition is found.
\end{proof}

It is therefore possible to construct a `first-order' distribution function by using Lemmas \ref{lemma:distributive} and \ref{lemma:zero}. By splitting $\Delta f$ into a linear and nonlinear contribution:

\begin{equation}
f = f|_{t=0} + \int\limits_{\Omega} \left\{ \dfrac{\delta f[\epsilon_0]}{\delta \epsilon_0(\vec{\mu}_j)} \Delta \epsilon \dif \vec{\mu}_j \right\} + \mathcal{O}[\Delta \epsilon^2]
\end{equation}

By discarding the terms which nonlinearly depend on $\Delta \epsilon$, one can approximate the distribution function. This generalises the technique used by Bernstein, Greene and Kruskal to obtain `first-order' distribution functions for so-called BGK modes.\cite{bernstein1957exact}

It is however also worth noting that under the approximation that $\Delta f$ is a linear functional, the functional is distributive under addition:

\[
f[\epsilon] \approx f[\epsilon_0] + f[\Delta \epsilon]
\]

\subsection{Gaussian expansion}
Hermite functions form an orthonormal basis, and have recently been explored in the literature by J. M. Heninger \emph{et al.} for G-transform based solutions of the linearised Vlasov-Poisson system. \cite{heninger2018integral} Gaussian functions have a similar form to Hermite functions, but are easier to work with due to the lack of preceding Hermite polynomials in the function definition. In addition, the previous subsection motivates the use of Gaussian functions via the ability to reduce to BGK theory.

In this subsection, we will examine solutions of the form $f = f_{(0)}(x,v,t) + \Delta f[\Delta \epsilon]$, under the following Gaussian expansion:

\begin{equation}
\Delta \epsilon(\chi,v,t) = \sum\limits_l W_l(\chi,t) \exp \left[ - \left(\dfrac{v-u_l(t)}{v_{N,l}(\chi,v,t)}\right)^2 \right]
\end{equation}

where $v_{N,l}(\chi,v,t)$ is a normalising function with units of velocity, and $W_l(\chi,t)$ is an enveloping function with units of energy. 

By inserting $\epsilon$ into equation \eqref{eq:epsilon_vlasov} one finds that to solve the Vlasov equation, $v_{N,l}$ is an approximate solution of the following differential equation under $x,v,t$ coordinates:

\begin{widetext}
\begin{equation}
\label{eq:master}
-2W_l \dfrac{(v - u_l)^2}{(v_{N,l})^3} \left[\pdif{}{t} +  v \pdif{}{x} - \dfrac{q}{m} \pdif{\phi}{x} \pdif{}{v} \right] v_{N,l} \approx 2(v-u_l) \dfrac{W_l}{(v_{N,l})^2} \left[\dfrac{\dif u_l}{\dif t} + \dfrac{q}{m} \pdif{\phi}{x}\right] + \left[\pdif{}{t} + v\pdif{}{x}\right] W_l
\end{equation}
\end{widetext}

For the case of a single frequency wave in the system, \eqref{eq:master} is exact. Otherwise, we have assumed that the overlap integral of two Gaussians is approximately zero. As such, for resonant structures which are close to each other in phase-space, the validity of \eqref{eq:master} breaks down.

It is possible to expand the left hand side of equation \eqref{eq:master} using a Laurent series:

\begin{equation}
\label{eq:laurent}
\renewcommand{\arraystretch}{2.5}
\begin{array}{r l}
&\displaystyle\sum\limits_{j=-\infty}^{\infty} L_l^{(j)}(x,t) (v - u_l)^j \\
&\hspace{20pt} = -2W_l \dfrac{(v - u_l)^2}{(v_{N,l})^3} \left[\pdif{}{t} +  v \pdif{}{x} - \dfrac{q}{m} \pdif{\phi}{x} \pdif{}{v} \right] v_{N,l}
\end{array}
\end{equation}

where $\{L_l^{(j)}\}$ are Laurent coefficients. Accordingly, one finds that for the analytic part of equation \eqref{eq:master}, in the limit that $v \to u_l$:

\begin{equation}
\label{eq:Aj}
L_l^{(0)}(x,t) \approx \left[\pdif{}{t} + u_l\pdif{}{x}\right] W_l
\end{equation}

From this point onward, we will seek solutions where $L_l^{(0)}$ is the only term in the Laurent expansion.

\begin{lemma}
\label{lemma:vns}
If the Laurent series expansion of the left hand side of equation \eqref{eq:master} only has one term, $L_l^{(0)} (v - u_l)^0$, then, equation \eqref{eq:master} is approximately solved with:

\[
v_{N,l}^2 \approx - \dfrac{2 W_l}{\partial_x W_l} \left[\dfrac{\dif u_l}{\dif t}+ \dfrac{q}{m} \pdif{\phi}{x}\right]
\]
\end{lemma}

\begin{proof}
Using the expansion given by equation \eqref{eq:laurent}, equation \eqref{eq:master} takes the form:

\[
\renewcommand{\arraystretch}{2.5}
\begin{array}{r l}
L_l^{(0)}(x,t) &\approx 2(v-u_l) \dfrac{W_l}{(v_{N,l})^2} \left[\dfrac{\dif u_l}{\dif t} + \dfrac{q}{m} \pdif{\phi}{x}\right] \\
& \hspace{40pt} + \left[\pdif{}{t} + v\pdif{}{x}\right] W_l
\end{array}
\]

By substituting equation \eqref{eq:Aj} the lemma is proved if the trivial solution of $(v - u_l) \neq 0$ is ignored.
\end{proof}

\subsection{Example solutions}
Here, we will examine the possible solution:

\begin{equation}
W_l(x,t) = -q [\phi_l(x,t) + \tilde{W}_l(t)]
\end{equation}

where $\phi = \sum\limits_l \phi_l(x,t)$. If each $\phi_l$ has the form (from equation \eqref{eq:wavepacket}):

\[
\phi_l(x,t) \equiv |\phi_{ll}| \exp\left[\int\limits_0^t \gamma_{ll} \dif \tau\right] \cos \left[\chi_{ll} + \theta_l\right]
\]

where $\theta_l$ is some initial phase, then from \lemref{lemma:vns}:

\[
v_{N,l}^2 \approx \dfrac{2 W_l \left[\dfrac{\dif u_l}{\dif t}+ \dfrac{q}{m} \pdif{\phi}{x}\right]}{k_l q\phi_{ll} \exp\left[ \int\limits_0^t \gamma_{ll} \dif \tau\right] \sin \left[\chi_{ll} + \theta_l\right]}
\]

We desire $\Delta \epsilon$ to be smooth, finite, and real-valued everywhere, which in turn requires $0 \leq v_{N,l} < \infty$. To enforce $|v_{N,j}| < \infty$ everywhere:

\begin{equation}
\label{eq:cancel}
\left\{W_l \left[\dfrac{\dif u_l}{\dif t}+ \dfrac{q}{m} \pdif{\phi}{x}\right]\right\}_{\chi_{ll} = - \theta_l + n \pi} = 0
\end{equation}

where $n$ is an integer. This can be satisfied in two ways. Equation \eqref{eq:cancel} can be satisfied if:

\[
\left.
\renewcommand{\arraystretch}{2.4}
\begin{array}{r l}
\tilde{W}_l &\displaystyle= -|\phi_{ll}| \exp\left[\int\limits_0^t \gamma_{ll} \dif \tau\right] \\
0 &\displaystyle= \dfrac{\dif u_l}{\dif t} + \left.\dfrac{q}{m} \pdif{\phi}{x}\right|_{\chi_{ll} = -\theta_l + n\pi}
\end{array} \right\} \, \textrm{for} \, W_l \geq 0
\]

However, one cannot guarantee that this is true for each and every $n$. Therefore, we use the alternative solution:

\[
\left.
\renewcommand{\arraystretch}{2.4}
\begin{array}{r l}
\tilde{W}_l &\displaystyle= |\phi_{ll}| \exp\left[\int\limits_0^t \gamma_{ll} \dif \tau\right] \\
0 &\displaystyle= \dfrac{\dif u_l}{\dif t} + \left.\dfrac{q}{m} \pdif{\phi}{x}\right|_{\chi_{ll} = -\theta_l}
\end{array} \right\} \, \textrm{for} \, W_l \leq 0
\]

such that $\tilde{W}_l$ is the amplitude of the wave with phase velocity $u_l$. Therefore, by examining the rate of change of $u_l$, one finds the equations:

\begin{widetext}
\begin{equation}
\label{eq:matrix}
\dfrac{\dif \omega_{ll}}{\dif t} \approx 
\sum\limits_{j} \dfrac{qk_j k_l}{m} \tilde{W}_{j}(t) \sin \left[\left\{\int\limits_0^t \left(\omega_{jj} - \dfrac{k_j}{k_l}\omega_{ll}\right) \dif \tau - \left(\theta_{j} - \dfrac{k_j}{k_l}\theta_l\right)\right\}\right] \,\, \forall \, l
\end{equation}
\end{widetext}

The system of equations \eqref{eq:matrix} must be consistent with the solution for $\{p_{sl}(t)\}$ given by \eqref{eq:psl}, determining the evolution of $\phi$.

\subsubsection{Stationary frequency}
\label{sec:static}
For the case of a single wave in the system, equation \eqref{eq:matrix} yields a stationary frequency wave. This exact solution is consistent with equation \eqref{eq:gsl} if the growth rate is also static:

\begin{equation}
\label{eq:wll}
\dfrac{\dif p_{ll}}{\dif t} = 0
\end{equation}

Therefore, we find:

\[
v_{N,l}^2 \approx \dfrac{2 q [\phi_l(x,t) + \tilde{W}_l(t)]}{m}
\]

This solution therefore gives to lowest order in $(v-u_l)$:

\[
\Delta \epsilon \approx \dfrac{1}{2} m (v-u_l)^2 - q (\phi_l + \tilde{W}_l)
\]

This gives an estimate for the separatrix width a factor of $\sqrt{2}$ higher, as the particle acts as though it is trapped in a potential well twice as large as the real potential. 

In contrast, if one chose to use $W_l \to - W_l$ instead, $v_{N,l}^2$ would be negative, allowing for us to recover the BGK solution. In such a case, for the full Gaussian representation $\Delta \epsilon$ does not vanish as $v \to \infty$ (as is also the case for BGK modes), and moments of the $\epsilon$ are not real valued (as $v_{N,l}$ is imaginary).

In \figref{fig:linear}, we show the corresponding full solution for $\Delta \epsilon(x,v,t=0)$ under this approximation, and how it deviates from BGK theory away from resonance. Here, we use normalised values of $|\phi_{ll}| = q = m = 1$, and $k = 2 \pi / L$.

\subsubsection{Sweeping frequency}
For the case of two waves in the system with the same wavenumber:

\begin{equation}
\label{eq:wlg}
\dfrac{\dif}{\dif t} \left(\begin{array}{c} \omega_+ \\ \omega_- \end{array}\right) \approx \dfrac{qk^2}{m} \sin \left[ \left\{\int\limits_0^t \tilde{\omega} \dif \tau - \tilde{\theta}\right\} \right] \left(\begin{array}{c} \tilde{W}_- \\ -\tilde{W}_+\end{array}\right)
\end{equation}

where $\tilde{\omega} \equiv \omega_+ - \omega_-$, and $\tilde{\theta} \equiv \theta_+ - \theta_-$. If one examines the case where:

\[
\dfrac{\dif \omega_+}{\dif t} \approx - \dfrac{\dif \omega_-}{\dif t}
\]

such that the frequency sweep is roughly symmetric, then $\tilde{W}_+ \approx -\tilde{W}_-$. This type of mode therefore shares similarities with symmetrically sweeping modes as is observed in experiments on devices such as the Mega Amp Spherical Tokamak (MAST) \cite{gryaznevich2006perturbative}. For either of the waves to be linearly stable, we require $\exists \, t_0: \gamma_{\pm}(t = t_0) = 0$. However, the waves can still be nonlinearly unstable. To assess the nonlinear stability, we combine equations \eqref{eq:gamma_ev} and \eqref{eq:omega_ev}:

\[
\dfrac{\dif \gamma_{\pm}}{\dif t} = - \gamma_{\pm}^2 + \left(\dfrac{\dot{\omega} + \Im(\Gamma_{\pm}^2)}{2 \gamma_{\pm}}\right)^2 + \Re(\Gamma_{\pm}^2)
\]

where $\dot{\omega} = \pm \dif \tilde{\omega} / \dif t$. For nonlinear stability, we require that $\dif \gamma_{\pm} / \dif t < 0$. Therefore:

\[
\lim_{t \to \infty} [\dot{\omega} + \Im(\Gamma_{\pm}^2)]^2 < \gamma_{\pm}^2[\gamma_{\pm}^2 -  \Re(\Gamma_{\pm}^2)]
\]

Additionally, for stability we require that $\lim_{t \to \infty} \dot{\omega} = 0$. Therefore, for nonlinear stability the following inequality must be satisfied:

\begin{equation}
\label{eq:ineq}
\gamma_{\pm}^2 > \dfrac{1}{2} \left[\Re(\Gamma_{\pm}^2) + \sqrt{\Re(\Gamma_{\pm}^2)^2 + 4\Im(\Gamma_{\pm}^2)^2}\right]
\end{equation}

where one can identify that $\Gamma_{\pm}^2 \equiv \lim_{t \to \infty} p_{\pm}^2$. In \figref{fig:sketch}, we sketch the regions defined by \eqref{eq:ineq} with the corresponding sign of $\dif \gamma_{\pm} / \dif t$. Intersections between the line $\gamma_{\pm}$ and $\lim_{t \to \infty} \gamma_{\pm}$ show stable values for the growth rate, corresponding to nonlinearly stable (metastable limit with $\gamma_{\pm} < 0$) and nonlinearly unstable (stable limit with $\gamma_{\pm} > 0$) states.

\begin{figure}[h!!]
    \includegraphics[width=0.45\textwidth]{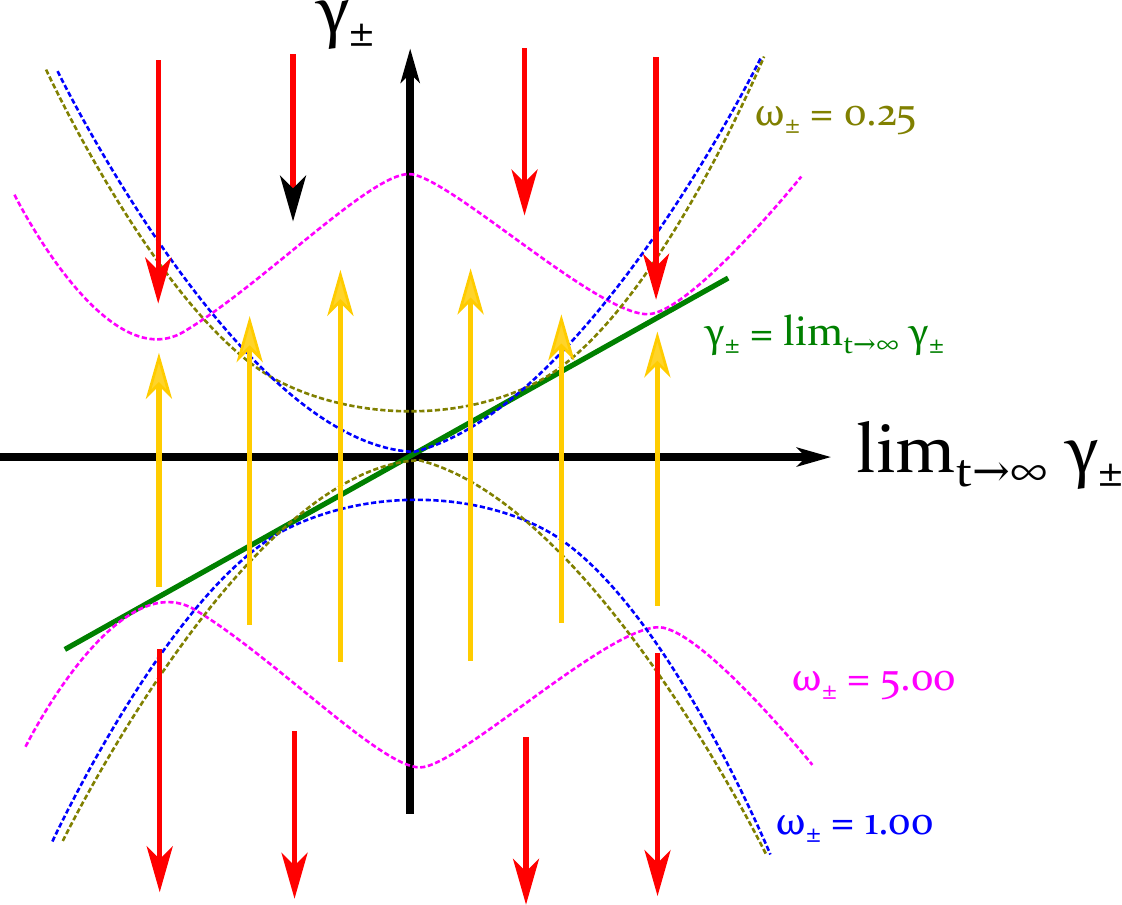}
    \caption{Sketches illustrating boundaries defined by \eqref{eq:ineq}, for $\omega_{\pm} \in {0.25,1,5}$. Yellow arrows and red arrows denote $\dif \gamma_{\pm} / \dif t > 0$ and $\dif \gamma_{\pm} / \dif t < 0$ respectively for $\omega_{\pm} = 5$. Points where the boundaries touch the line $\gamma_{\pm} = \lim_{t \to \infty} \gamma_{\pm}$ indicate stable or metastable values of $\gamma_{\pm}$ for $\gamma_{\pm} > 0$ and $\gamma_{\pm} < 0$ respectively. The stable value therefore corresponds to a nonlinearly unstable solution, while the metastable value corresponds to a nonlinearly stable solution.}
    \label{fig:sketch}
\end{figure}

Everywhere else above the line defined by the negative root of \eqref{eq:ineq}, the wave alternates between increasing and decreasing growth rates. As the choice of initial conditions influences $\Gamma_{\pm}$, some initial conditions will not have a limit for $\gamma_{\pm}$. This allows for either repeated chirping as observed in simulations and experiments \cite{vann2003fully,vanzeeland2019alfven}, or nonlinear instability (if $\int_{t_1}^{t_2} \gamma_{\pm} \dif t > 0$ for $t_2 > t_1$). This may be a candidate for the rapid frequency chirping observed during abrupt large events or mode avalanching in tokamaks \cite{sharapov2002alfven,gryaznevich2006perturbative,podesta2009experimental}. 

Between $\gamma_{\pm} = 0$ and the line defined by the negative root of \eqref{eq:ineq}, the growth rate is negative, but decreasing in amplitude. Waves here are linearly stable, but may become nonlinearly unstable if they are able to cross the $\gamma_{\pm} = 0$ line with finite amplitude. Below the line defined by the negative root of \eqref{eq:ineq}, the growth rate is negative and always decreasing, corresponding to both linear and nonlinear stability.

From equation \eqref{eq:wlg}, we find that the magnitude of the sweeping rate is directly proportion to $\tilde{A}$. Therefore, if there is a large difference in the wave amplitudes, the system here is nonlinearly unstable.

It is worth noting that this solution is only approximately valid for structures which have a very small overlap integral in phase-space. As such, one can consider these long-range, coupled structures in a approximate fashion by evolving a continuum of superposed, stationary frequency solutions.

\subsection{Separatrix}
We define the separatrix as the largest closed contour in phase-space. Here, we examine only a single value of $l$ active in the system. Here, we examine $\Delta \epsilon(\chi,v,t) = c$, an unspecified constant. Points on such a contour only exist for:

\begin{equation}
0 \leq \dfrac{c}{W_l} < 1 
\end{equation}

as $\exp(-z) \in (0,1]$ for positive $z$. As such, if $c < W_l$, no point exists on the contour for the corresponding values of $\chi,t$. Therefore, closed contours exist for $c > \textrm{min}(W_l)$, and the separatrix is given by $c = 0$. At this fixed value of $\Delta \epsilon(\chi,v,t)$:

\[
v_{\textrm{sep.}} = u_l \pm \lim_{c \to 0} \sqrt{-\dfrac{2 W_l}{m} \ln \left( \dfrac{W_l}{c}\right)}
\]

As a result, all the particles here are `trapped'. However in reality, weakly bound particles would scatter out of the potential via neoclassical transport in a model that considers collisions.

In light of this, we use the following fit for an effective separatrix to compare to BGK theory:

\[
v_{\textrm{sep.}} \approx u_l \pm \sqrt{-\dfrac{2 W_l}{m}}
\]

which is given when the structure is at $1/\textrm{e}$ height. In comparison to BGK theory, our value of $W_l$ has twice the amplitude of $\phi_l$. Therefore, the width of the separatrix given by this theory is roughly a factor of $\sqrt{2}$ larger than that which is expected for a BGK island.

\section{Conclusions}
In conclusion, we have solved a variety of problems related to 1D kinetic plasmas, allowing for a detailed description of frequency chirping and basis expansion of the distribution function.

We have shown that a Berk-Breizman sink that linearly dissipates electric field energy from a 1D kinetic plasma can be approximated as a Krook-like collision operator. We have shown in \secref{sec:nonlinear} that it is possible to represent the time-varying frequency of waves in the system via the corresponding growth rates. It was therefore shown in \secref{sec:effective} that it is possible to analytically produce models that account for chirping and nonlinear saturation.

We have shown that the Vlasov equation can be exactly solved for the general case of varying complex frequency waves, and show that a Gaussian basis expansion allows one to find an exact solution for a single wave in the system, and to recover approximate solutions.

We find that under a Gaussian expansion for $\Delta \epsilon$, the separatrix width in velocity for a single mode is a factor of $\sqrt{2}$ larger than the width expected for a BGK island. The family of modes we derive therefore exhibit highly nonlinear orbits, where passing particles for a BGK wave are here interacting with the wave. This implies that mode overlap occurs for modes which are further apart in phase velocity to one another than arises from BGK theory. However, the modes found here may not typically be excited, the island width observed in simulations is typically that which is expected from BGK theory \cite{berk1997spontaneous,meng2018resonance}. 

While not shown here, a logical next step would be to use this formalism to investigate the full Laurent expansion in equation \eqref{eq:laurent}, allowing for a greater range of acceptable solutions. Alternatively, using a Hermite expansion for $\epsilon$ would allow for one to find an exact form of equation \eqref{eq:master} which allows one to examine frequency bifurcation.

\section{Acknowledgements}
The author was funded by the EPSRC Centre for Doctoral Training in Science and Technology of Fusion Energy grant EP/L01663X. This work has been carried out within the framework of the EUROfusion Consortium and has received funding from the Euratom research and training programme 2014-2018 under grant agreement No 633053. The views and opinions expressed herein do not necessarily reflect those of the European Commission. 

The author would like to thank V. N. Duarte, Z. Lin and H. Hezaveh for useful discussions. The author would also like to thank K. Imada and R. G. L. Vann for valuable help with improving this manuscript.

\small
\bibliographystyle{unsrt}
\bibliography{biblio}

\end{document}